\documentclass{article}

\usepackage{amsmath}
\usepackage{amssymb}
\usepackage{amsthm}
\usepackage{array}
\usepackage{enumerate}
\usepackage{layout}
\usepackage{graphicx}
\usepackage{fancyhdr}

\newtheorem{theorem}{Theorem}[section]
\newtheorem{lemma}[theorem]{Lemma}

\begin{document}
\author{Jeffrey Ling, Kai Xiao, Dai Yang}
\title{Online Algorithms Modeled After Mousehunt}

\maketitle
\noindent
\thispagestyle{fancy}

\begin{abstract}
In this paper we study a variety of novel online algorithm problems inspired by the game Mousehunt. We consider a number of basic models that approximate the game, and we provide solutions to these models using Markov Decision Processes, deterministic online algorithms, and randomized online algorithms. We analyze these solutions' performance by deriving results on their competitive ratios.
\end{abstract}

\section{Introduction}

Mousehunt is a Facebook game developed in 2006 by HitGrab Inc. The goal of the game is to catch mice using a variety of traps. Each species of mice is worth a certain amount of points and gold. Although collecting gold helps the player afford better traps, collecting points is the ultimate goal of the game. One particularly focus-worthy aspect of the game is that Ronza, a nomadic merchant, visits for short periods of time roughly once a year, and sells valuable exclusive traps during these unannounced visits.

We introduce simple models of this game that involve optimizing the number of points gained over a finite time interval. While the problem's overall description will resemble the classic ski rental problem, the finer details will differ, and we will be able to show different lower bounds on the competitive ratio. We will approach these problems using both deterministic and randomized online strategies to try to achieve the best possible competitive ratios.

In this paper, we will use the convention that the competitive ratio $r$ is always less than 1, i.e. if our algorithm earns value $C_A$, and the optimal offline algorithm earns value $C_{OPT}$, then $C_A \geq r \cdot C_{OPT}$.

We begin by proposing a simple model for Mousehunt, where we start with a basic trap that can selectively catch mice worth one point or one gold. Assuming that we don't know when Ronza will arrive next, and that we have some estimate $x$ of the benefit we gain from Ronza's traps, we are able to prove that it is optimal to hunt for gold if and only if the ratio of the gold cost of the trap $c$ to the timespan $T$ satisfies $c/T \leq 1-\frac{1}{\sqrt{x}}$. That is, the potential amount of benefit we can gain is worth it iff the cost of the trap is not too high. In this case, we obtain a competitive ratio of $1/\sqrt{x}$.

If we randomize our strategy, it turns out we can do better than $1/2$-competitive on average.

\section{A Simple Model}

\subsection{Modeling Mousehunt} In the game Mousehunt, players attempt to catch different types of mice which give players rewards of points and gold when caught. The reward can therefore be represented as a vector $(p,g)$ where $p$ is the amount of points and $g$ is the amount of gold. In the real game, players can use different trap setups to increase their catch rates against certain sets of mice, and the player can select which mice to target by changing their trap setup (by arming certain types of cheese or traveling to certain locations). Traps can be purchased for certain amounts of gold from trapsmiths or from Ronza, who visits the mouse-hunting land of Gnawnia once a year.

\subsection{Initial Problem} Suppose that there are three traps available, $A_1$, $A_2$, and $B$. Using trap $A_1$ will allow the player to catch a mouse worth $1$ gold, and using trap $A_2$ will allow the player to catch a mouse worth $1$ point. Using trap $B$ will allow the player to catch a mouse worth $x$ points for $x>1$. Traps $A_1$ and $A_2$ are available at the start of the game while trap $B$ can only be purchased from the wandering merchant Ronza. There are $T$ time steps in total. Ronza appears for a single time step at some unknown time $y$, and sells trap $B$ at cost $c$. Before each time step, and also at the end, the player makes a decision to lay down either trap $A_1$, trap $A_2$, or $B$, and then immediately reaps the rewards of their choosing. This setup poses two different but related problems: one, to determine an optimal trap-choosing strategy to maximize the expected payoff given some distribution assumptions, and two, to find a strategy that maximizes the competitive ratio.

\section{Markov Decision Process Analysis}

\subsection{Distribution and Optimality}
For the rest of this section, we will assume that Ronza's arrival time is distributed uniformly among $1,2,...,T$. We can perform similar analyses with different distributions, but it will be most clear to analyze the uniform distribution case.

We establish the following lemma concerning optimal online algorithms which solve this problem.
\begin{lemma}
The online algorithm, if it is optimal, can be assumed to take on one of the following formats:
\begin{itemize}
\item Only hunt for the mouse worth $1$ point.
\item Hunt for the mouse worth $1$ coin until $c$ coins are gathered or until the merchant arrives. Purchase the trap $B$ if possible when the merchant arrives, and only hunt for mice worth a maximum amount of points afterward.
\end{itemize}
\end{lemma}

\begin{proof}
First, suppose that the online algorithm $O_1$ at time $t$ decides to earn points and at the next time step $t+1$ decides to collect gold. Then, consider the alternative online algorithm $O_2$ that collects gold at $t$ and earns points at $t+1$. Both of these choices are conditioned on Ronza having not arrived yet, as it is clear that the only right decision after Ronza arrives is to earn points. If Ronza appears after time $t+1$ or before time $t$, both algorithms perform equally well. If Ronza appears between times $t$ and $t+1$, then the algorithm $O_1$ will have performed worse. Thus, any deterministic optimal online algorithm will collect gold at the beginning, if at all, and hence take on one of the two prescribed formats.
\end{proof}

\subsection{Markov Decision Process}

One approach for solving the initial problem is to view it as a Markov Decision Process. For a fixed $x$, define $f(c,T)$ to be the optimal decision when there are $T$ time steps remaining and we require $c$ more coins for Ronza's trap; the lemma tells us that $f(c,T)$ can only take on two forms: aim for points, or aim for $c$ coins. Define $g(c,T)$ to be expected payoff in points when pursuing this optimal decision. The base cases go as follows. When $c = 0$, we have $f(c,T)$ as aim for points and $g(c,T) = \frac{x+1}{2}(T+1)$, because Ronza comes in the middle of the time steps on average. When $c=1$, we have $g(c,T)$ as $\max(\frac{x+1}{2}T, T+1)$: the latter payoff is achieved by aiming for points, and the former payoff is achieved by aiming for a single coin. When $T = 1, c > 1$, we have $f(c,T)$ as aim for points and $g(c,T) = 2$. Then:

\begin{theorem}
In the state $(c,T)$ where $T > 1$ and $c > 0$, $f(c,T)$ and $g(c,T)$ can be determined by comparing $T+1$ and $\frac{T-1}{T}g(c-1,T-1) + 1$. If the latter is larger, then it is the expected payoff, and the optimal move $f(c,T)$ is to aim for $c$ gold. Otherwise, if the former is larger, then it is the expected payoff, and the optimal move $f(c,T)$ is to aim for points.
\end{theorem}

\begin{proof}
By the lemma, we may reduce the optimal decision to two possibilities: aim for points, or aim for $c$ coins. If we decide to aim for points, the payoff will be $1+g(c,T-1)$. However, it is not rational to aim for coins in future time steps if we do not aim for coins now. Hence, the payoff must evaluate to $1+g(c,T-1) = 1 + T$. Otherwise, if we decide to aim for $c$ coins, then our payoff will depend on whether or not Ronza arrives in the next time step.

If Ronza arrives in the next time step, then we have no choice but to gain $1$ point after all remaining time steps, accruing a total of $T$ points. This occurs with probability $\frac{1}{T}$.

If Ronza arrives after the next time step, then her arrival time will be uniformly distributed among the remaining $T-1$ time steps, just as in the state $(c-1,T-1)$. This reduction allows us to conclude that the payoff is $g(c-1,T-1)$, and occurs with probability $\frac{T-1}{T}$.

Thus, the expected payoff when aiming for $c$ coins is $\frac{T-1}{T}g(c-1,T-1) + \frac{1}{T}\cdot T = \frac{T-1}{T}g(c-1,T-1) + 1$. Between this payoff and the payoff of $T+1$ from greedily amassing points, whichever one is larger will dictate both $f(c,T)$ and $g(c,T)$.
\end{proof}

\begin{theorem}
Let $r = \frac{c}{T}$. Then asymptotically, $f(c,T)$ will dictate that it is optimal to catch mice to aim for $c$ gold iff $r \le 1-\frac{1}{\sqrt{x}}$.
\end{theorem}

\begin{proof}
We will no longer keep track of unneeded additive constants as we are determining asymptotic behavior. Suppose that Ronza's arrival time is $y$. We already know from the lemma that the alternative to catching mice to aim for $c$ gold is catching mice to aim for points. The latter yields a payoff of $T$ (in fact, $T+1$). Now we will compute the payoff of the former in a non-recursive fashion.

\begin {itemize}
\item Case 1: Ronza comes after time $c$.

This means that $c \leq y$, and decisions made according to this algorithm will earn $y-c+x(T-y)$ points: $y-c$ points from mice worth $1$ point each, and $x(T-y)$ points from mice worth $x$ points each. The average value of $y$ here is $\frac{c+T}{2}$.

\item Case 2: Ronza comes before time $c$.
This means $c>y$, and the player can never afford trap $B$, so the algorithm will simply catch mice worth $1$ point for all remaining time steps, yielding $T-y$ points. The average value of $y$ here is $\frac{c}{2}$.
\end{itemize}

This means that if we decide to aim for $c$ coins, our expected gain in points is $\frac{T-c}{T}\cdot (\frac{c+T}{2}-c+x(T-\frac{c+T}{2})) + \frac{c}{T}(T-\frac{c}{2})$. Equalizing the payoffs of the two decisions allows us to determine the asymptotic boundary:
\begin{equation*}
\begin{split}
T & = \frac{T-c}{T}\cdot (\frac{c+T}{2}-c+x(T-\frac{c+T}{2})) + \frac{c}{T}(T-\frac{c}{2}) \\
1 & = (1-r)\cdot(1-r)(\frac{x+1}{2}) + r\cdot(1-r/2) \\
0 & = \frac{x}{2}r^2 - xr + \frac{x-1}{2} \\
0 & = xr^2 - 2xr + x-1
\end{split}
\end{equation*}
Solving this quadratic and taking the root such that $r < 1$ yields $r = 1-\frac{1}{\sqrt{x}}$. Therefore, the asymptotic behavior of this Markov Decision Process can be summarized as: until Ronza makes an appearance, aim for $c$ coins if the state $(c,T)$ satisfies $\frac{c}{T} \le 1-\frac{1}{\sqrt{x}}$, and aim for points otherwise.
\end{proof}

Below, we present a few values of $g(c,T)$ when $x=2$, that is to say, when trap $B$ is twice as effective as trap $A$. Bolded entries represent states in which $f(c,T)$ dictates that aiming for $c$ coins is strictly better than aiming for points. Notice that the boundary between the two strategies being optimal closely follows the line $c = (1-\frac{1}{\sqrt{2}})T$, as shown in the theorem above.

\begin{center}
\begin{tabular}{l*{9}{c}r}
(c,T) & 1 & 2 & 3 & 4 & 5 & 6 & 7 & 8 & 9 & 10 \\
\hline
0     & \textbf{3}&	\textbf{4.5}&	\textbf{6}&	\textbf{7.5}&\textbf{9}&\textbf{10.5}&	\textbf{12}&	\textbf{13.5}&	\textbf{15}&	\textbf{16.5}   \\
1     & 2	&	3	&	\textbf{4.5}	&	\textbf{6}	&	\textbf{7.5}	&	\textbf{9}	&	\textbf{10.5}	&	\textbf{12}	&	\textbf{13.5}	&	\textbf{15}
  \\
2     & 2	&	3	&	4	&	5	&	6	&	\textbf{7.3}	&	\textbf{8.7}	&	\textbf{10.2}	&	\textbf{11.7} & \textbf{13.2}
  \\
3     & 2	&	3	&	4	&	5	&	6	&	7	&	8	&	9	&	10	&	\textbf{11.5}
  \\
4     & 2	&	3	&	4	&	5	&	6	&	7	&	8	&	9	&	10	&	11
  \\
5     & 2	&	3	&	4	&	5	&	6	&	7	&	8	&	9	&	10	&	11
  \\

\end{tabular}
\end{center}

It is not difficult to generalize this table for all $x$. Note that for a given $(c,T)$, we can use dynamic programming to fill in this table and compute $f(c,T)$ and $g(c,T)$ in $O(cT)$ time. However, for large values of $c,T$, we can bypass Markov Decision Processes altogether and use Theorem $3.3$ to compute them in $O(1)$ time.

Our analyses so far relied on the distribution of Ronza's arrival time being uniform. If the distribution is not known, however, then Markov Decision Processes cannot be used to model the decision problem. Next, we will show a solution to the algorithm online problem against any adversary.

\section{Competitive Analysis Model}

In this section we analyze online algorithms to solve the initial problem without a known probability distribution. We examine both an optimal deterministic online algorithm, and a randomized online algorithm that achieves the best possible asymptotic competitive ratio.

\subsection{Competitive Ratio Analysis: Deterministic}

If the online algorithm decided only to earn points, the worst case scenario is that the optimal offline algorithm was to earn $c$ gold and buy trap $B$ when Ronza arrived to earn $x$ points per subsequent mouse. This means that $c\leq y$, and the offline algorithm would garner $\max(y-c+x(T-y), T)$ points. If $\max(y-c+x(T-y), T) = T$, the competitive ratio is $1$. We now analyze the case when $\max(y-c+x(T-y), T) = (y-c)+x(T-y)$, where the competitive ratio of the two algorithms is $\frac{T}{(y-c)+x(T-y)}$.

We can find an upper bound on the reciprocal of this fraction to find a lower bound on this ratio. Let $r = \frac{c}{T}$. Then the reciprocal of the ratio can be upper bounded as follows.
\begin{align*}
\frac{(y-c)+x(T-y)}{T} &= x-\frac{c}{T}+y\left(\frac{1-x}{T}\right) \\
&\leq x-\frac{c}{T}+c\left(\frac{1-x}{T}\right)
= x\left(1-\frac{c}{T}\right) = x(1-r)
\end{align*}
where we used the fact that $(1-x)$ is negative and $c\leq y$. The competitive ratio is thus lower bounded by $\frac{1}{x(1-r)}$.

If the online algorithm decided to collect gold to try to buy the trap, the worst case scenario is that Ronza appears before the algorithm has collected enough gold to purchase the trap. In this case, $c > y$, and the online algorithm would earn $T-y$ points (it would earn gold until the merchant arrived at time $y$, after which it would earn points) and the offline optimal algorithm would earn $T$ points. The competitive ratio here is $\frac{T-y}{T}$. Here we find that the ratio is
\begin{equation*}
1-\frac{y}{T} > 1-\frac{c}{T} = 1-r
\end{equation*}
Thus, given knowledge of $x,c,T$, we can choose our strategy based on the larger of the values $\frac{1}{x(1-r)}$ and $1-r$ to achieve the best competitive ratio. Since these two expressions vary in opposite directions as $r$ increases, the worst case occurs when the two expressions are equal, i.e.
\begin{equation*}
\frac{1}{x(1-r)}=(1-r) \iff 1-r=\frac{1}{\sqrt{x}}
\end{equation*}
Thus, since our analysis was worst case, we obtain a tight lower bound of $\frac{1}{\sqrt{x}}$ on our competitive ratio. Thus this deterministic algorithm is always $\frac{1}{\sqrt{x}}$ competitive for all values of $r$.

\subsection{Competitive Ratio Analysis: Randomized}
Here we describe and analyze a randomized online algorithm that achieves a worst case competitive ratio of $1/2$ for all values of $x$, $c$, and $T$.

First, if $x(T-c) < T$, then it is always optimal just to go for points because going for the trap would result in $y-c+x(T-y) \leq x(T-c) < T$ points at the end. This information is available to the online algorithm so the online algorithm is exactly the same as the offline one in this case.

Thus, we just have to analyze the case that $x(T-c) \geq T$. As before, any optimal online algorithm will either only aim for points or it will collect gold until Ronza arrives or until it has collected $c$ gold, after which it will collect points and buy Ronza's trap if possible.

The randomized algorithm is to choose to try to gather gold for Ronza's trap with probability $q$ and to choose to aim for points only with probability $1-q$. We will determine $q$ later based on $x$ and $r = \frac{c}{T}$. Call the gold gathering strategy $S_g$ and the point gathering strategy $S_p$. Let $S$ denote the randomly chosen strategy.

Given our choice of $q$, the adversary's goal is to choose Ronza's arrival time $y$ so that the competitive ratio is minimized. Our goal is to choose a $q$ to maximize this minimal value.

If $y < c$ then the competitive ratio is
$$
\begin{cases}
\frac{T-y}{T} &\mbox{if } S = S_g \\
1 &\mbox{if } S = S_p
\end{cases}
$$

If $y \geq c$ then the competitive ratio is
$$\begin{cases}
1 &\mbox{if } S = S_g \\
\frac{T}{(y-c)+x(T-y)} &\mbox{if } S = S_p
\end{cases}$$

Thus the overall competitive ratio, denoted by $R(x,c,T)$, is

\[R(x,c,T) = \left\{
  \begin{array}{lr}
    q \cdot \left(\frac{T-y}{T}\right)+(1-q) & : y < c\\
    q+(1-q)\frac{T}{(y-c)+x(T-y)}  & : y \ge c
  \end{array}
\right.
\]

Suppose $q$ has been chosen already. If the adversary chooses $y<c$, then $R(x,c,T)$ is minimized when $y$ is maximized, or when $y = c-1 \approx c$. If $y \geq c$, then minimizing $R(x,c,T)$ involves maximizing
\begin{equation*}
(y-c)+x(T-y) = xT-c+y(1-x)
\end{equation*}
which occurs when $y$ is minimized since $x > 1$. Thus in this case the adversary will choose $y = c$.

Thus we get that

\[R(x,c,T) = \left\{
  \begin{array}{lr}
    q \cdot\left(\frac{T-c}{T}\right)+(1-q) = 1-\frac{c}{T}q & : y < c\\
    q+(1-q)\frac{T}{x(T-c)} = \frac{T}{x(T-c)}+\left(1-\frac{T}{x(T-c)}\right)q  & : y \ge c
  \end{array}
\right.
\]

Note that the coefficient of $q$ is negative in the case $y<c$ and the coefficient of $q$ is positive in the case $y\geq c$  since $x(T-c) \geq T$. As $q$ increases, $R(x,c,T)$ goes down in the case $y < c$ and goes up in the case $y \geq c$.

The two possible expressions for $R(x,c,T)$ thus change in opposite directions as $q$ is varied. Thus, if we want to maximize the minimum of these two numbers, we have to set them equal. This gives
\begin{align*}
\frac{T}{x(T-c)}+\left(1-\frac{T}{x(T-c)}\right)q &= 1-\frac{c}{T}q \\
\iff q\left(1-\frac{T}{x(T-c)}+\frac{c}{T}\right) &= 1-\frac{T}{x(T-c)} \\
\iff q\left(\frac{x(T-c)(T)-T^2+cx(T-c)}{x(T-c)(T)}\right) &= \frac{-T+x(T-c)}{x(T-c)} \\
\iff q &= \frac{-T^2+xT(T-c)}{x(T-c)(T)-T^2+cx(T-c)} \\
q &= \frac{xT^2-T^2-cxT}{xT^2-T^2-c^2x}
\end{align*}

This gives the optimal value of $q$ given $x$, $c$, and $T$. If we let $r=\frac{c}{T}$, we can rewrite this as
\begin{equation*}
q = \frac{x-1-rx}{x-1-r^2x}
\end{equation*}
We can also rewrite the condition $x(T-c)\geq T$ as
\begin{align*}
x(1-r) \geq 1 \\
\iff x-xr \geq 1 \\
\iff \frac{x-1}{x} \geq r
\end{align*}

Plugging this value of $q$ back into $R(x,c,T)$, we get that the competitive ratio of this randomized algorithm is
\begin{equation*}
R(x,c,T) = 1-\frac{c}{T}q = 1-rq = 1-\frac{xr-r-r^2x}{x-1-r^2x} = \frac{(x-1)(1-r)}{x-1-r^2x}
\end{equation*}

To find a lower bound on the competitive ratio of $R$, we need to minimize $R(x,c,T)$ over all values of $x$, $c$, and $T$. Since $R$ only depends on the ratio $r = \frac{c}{T}$ asymptotically, we can just consider $R(x,r)$.

To minimize $R$, we compute the partial derivative
\begin{align*}
\frac{\partial R}{\partial r} &= \frac{-(x-1)(x-1-r^2x)-(-2rx)(x-1)(1-r)}{(x-1-r^2x)^2} \\
&= \frac{(x-1)(-(x-1-r^2x) + 2rx - 2r^2x)}{(x-1-r^2x)^2} \\
&= \frac{(x-1)(-r^2x + 2rx - x + 1)}{(x-1-r^2x)^2}
\end{align*}

$R$ is minimized when the partial derivative is $0$, or at the endpoints. At the endpoint $r=0$ we get $R(x,r) = 1$, and at the endpoint $r=\frac{x-1}{x}$, we know $x-1=rx$, so
\begin{equation*}
R(x,r) = \frac{rx(1-r)}{rx-r^2x} = 1
\end{equation*}

When the partial derivative is 0, we have that \begin{equation*}
(x-1)(1 - x(1-r)^2)=0
\end{equation*}
Since $x>1$, solving this gives
\begin{equation*}
r = 1-\frac{1}{\sqrt{x}}
\end{equation*}
This value of $r$ is valid because we know that $r=1-\frac{1}{\sqrt{x}}<1-\frac{1}{x} = \frac{x-1}{x}$.
Plugging this value of $r$ back into $R(x,r)$ we get that
\begin{equation*}
R(x,r) = \frac{(x-1)\left(1-\left(1-\frac{1}{\sqrt{x}}\right)\right)}{x-1-\left(1-\frac{1}{\sqrt{x}}\right)^2x} = \frac{1}{2}\left(\frac{1}{\sqrt{x}}+1\right)
\end{equation*}

We can find the minimum value of this over all $x$. If we let $x$ range from $1$ to $\infty$ we see that $R(x,r) > \frac{1}{2}$ for all $x$.

One way to intuitively derive the same competitive ratio is that the optimal offline algorithm clearly either goes for gold for Ronza's trap or goes for points the entire time. Thus, a basic randomized algorithm would just flip a fair coin. If the coin landed heads, it would choose to go for gold, and if it landed tails, it would choose to go for points. $\frac{1}{2}$ of the time the randomized algorithm will match the optimal offline algorithm, so we are at least $\frac{1}{2}$ competitive.

The result we showed was a proof that the algorithm was $\frac{1}{2}$ competitive for all values of $x$, $c$, and $T$. However, in reality, for most values of $x$, $c$, and $T$, the competitive ratio our randomized algorithm obtains is $R(x,r)$, which, for most values of $r$, is much better than the ratio obtained by single randomized coin flip algorithm. As an example, consider the graphs shown in Figure~\ref{fig:graph} of $R(x,r)$ for the cases $x = 4$ and $x = 100$ for the appropriate range of $r\in\left(0,\frac{x-1}{x}\right)$.

\begin{figure}[h]
\centering
\includegraphics[scale=0.5]{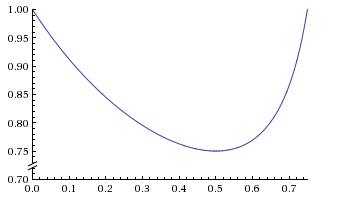}
\includegraphics[scale=0.5]{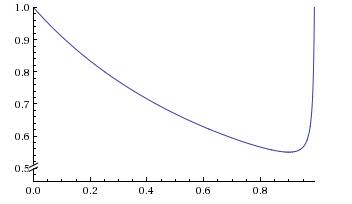}
\caption{Left plot: $x=4$. Right plot: $x = 100$.}
\label{fig:graph}
\end{figure}

\subsection{Proof of Optimality of Randomized Algorithm}
We will conclude this section with a proof that for any randomized online algorithm, an asymptotic competitive ratio of $\frac{1}{2}$ over all values of $x$, $c$, and $T$ is optimal using Yao's minimax principle.

Yao's minimax principle states that the expected cost of a randomized algorithm on the worst case input is no better than the worst case probability distribution of inputs for a deterministic algorithm that works best for that worst case distribution. Thus, we simply have to present a probability distribution such that no deterministic algorithm can perform very well for it.

It is not difficult to find such a distribution. Given the parameters $x$, $c$, and $T$, we can define two possible inputs (which are just values of $y$, Ronza's arrival time) and give them as inputs to the deterministic algorithm with probability $\frac{1}{2}$ each. The first input is $y = c$ and the second input is $y = c-1$. Then, the deterministic algorithm can be analyzed as follows.

\begin{itemize}
\item Case 1: The deterministic algorithm gathers coins for the first $c$ steps (assuming Ronza doesn't appear). Then, for the input $y=c$, the deterministic algorithm will be $1$-competitive. For the input $y=c-1$, at time $t=c-1$ the algorithm will realize that it can not buy the trap, and so it will then gather points afterwards. This gives us an overall competitive ratio of
\begin{equation*}
\frac{1}{2}\left(1+\frac{T-c}{T}\right)=\frac{1}{2}(1+(1-r))
\end{equation*}

\item Case 2: The deterministic algorithm does anything else.
Then, this means at time $t=c$, the algorithm will have less than $c$ gold. Then, for the input $y=c$, the algorithm will fail to buy the trap. For $y=c-1$, the algorithm will also fail to buy the trap, and at best it will be $1$-competitive. The algorithm will have earned at most $c$ points in the first $c$ time steps. The optimal offline algorithm will buy the trap for $c$ gold in the case $t=c$. Thus the overall competitive ratio is
\begin{equation*}
\frac{1}{2}\left(\frac{c+T-c}{(y-c)+x(T-y)}+1\right) \leq \frac{1}{2}\left(\frac{T}{x(T-c)}+1\right) = \frac{1}{2}\left(1+\frac{1}{x(1-r)}\right)
\end{equation*}
\end{itemize}

In that special case that $r=1-\frac{1}{\sqrt{x}}$, we know that the competitive ratio of this deterministic algorithm will be equal to $\frac{1}{2}\left(1+\frac{1}{\sqrt{x}}\right)$. Thus, as $x$ approaches infinity, any optimal deterministic algorithm for this input will be at best $\frac{1}{2}$ competitive, so $\frac{1}{2}$ competitive is as good as any randomized algorithm can be for the original problem.

\section{Unknown Cost, Fixed Arrival Time}
In this section, we explore a variant of our initial problem. Now, suppose that Ronza's arrival time $y$ is known, but the cost $c$ of buying trap $B$ is unknown.

\subsection{Online and Offline Algorithms}

First, we consider the offline algorithm. If $c$ is known, then this reduces to our initial offline problem. Therefore, if $c \leq y$, the optimal algorithm will get $\max(y-c + x(T-y), T)$ points, and if $c > y$, then the algorithm gets $T$ points.

Any deterministic online algorithm can be characterized by a number $m \leq y$, where it collects gold for $m$ of the first $y$ time steps, and hunts for points otherwise. It will buy Ronza's trap at time $y$ if it is affordable at time $y$. Now, there are two possible outcomes:
\begin{itemize}
\item If the algorithm manages to collect $c$ gold, i.e. $m \geq c$, then it is always optimal to buy trap $B$ if possible. Then the algorithm earns a total of $y-m + x(T-y)$ points.
\item If $m < c$, then the algorithm earns $T-m$ points.
\end{itemize}

In order to find a competitive ratio for this online problem, we consider the problem from the perspective of the adversary. For given choices of $c$ and $m$, we have the following cases and known ratios for the online and offline algorithms:
\begin{enumerate}
\item If $c > y$, we have a ratio of $\frac{T-m}{T}$
\item If $m < c \leq y$, and $y-c+x(T-y) < T$, we have a ratio of $\frac{T-m}{T}$
\item If $m < c \leq y$, and $y-c+x(T-y) \geq T$, we have a ratio $\frac{T-m}{y-c+x(T-y)}$
\item If $m \geq c$, and $y-c+x(T-y) < T$, we have a ratio of $\frac{y-m+x(T-y)}{T}$
\item If $m \geq c$, and $y-c+x(T-y) \geq T$, we have a ratio of $\frac{y-m+x(T-y)}{y-c+x(T-y)}$
\end{enumerate}
As the adversary, our job is to design $c$ for fixed $m$ in order to produce the smallest ratio. Note that since cases 1 and 2 produce the same ratio, and there always exists a $c$ to satisfy case 1, we need not consider case 2. Similarly, because
\begin{equation*}
y-m + x(T-y) \geq T-m \iff (x-1)(T-y) \geq 0
\end{equation*}
is true for any $m$, case 4 is redundant given case 1. 

Now notice that as the ratio in case 5 is an increasing function of $c$, its minimum is achieved at $c=0$. Thus, if we compare the ratios in cases 1 and 5, we have
\begin{equation*}
\frac{T-m}{T} < \frac{y-m+x(T-y)}{y+x(T-y)} \\
\iff Ty + xT^2 - xTy - my - xTm + mxy < Ty - Tm + xT^2 - xyT \\
\end{equation*}
Canceling terms and rearranging, we obtain
\begin{equation*}
0 < m(x-1)(T-y)
\end{equation*}
which is always true. Thus, case 5 is redundant given case 1. 

It remains to compare the ratios of cases 1 and 3, and for an online algorithm's choice of $m$, it is the adversary's goal to choose $c$ to obtain the smaller ratio. Note that the value of $c$ that minimizes ratio 3 is $c=m+1$, assuming such a $c$ satisfies the constraints, so we are left with two ratios to compare: $\frac{T-m}{T}$ from case 1, and $\frac{T-m}{y-m-1+x(T-y)}$ from case 3.

In words, case 1 corresponds to when both offline and online algorithm cannot buy Ronza's trap, and case 3 corresponds to when the offline can and does buy the trap while the online cannot, with conditions $c \leq (x-1)(T-y)$ and $m < c \leq y$.

\subsection{Worst Case Analysis}

Suppose the online algorithm chooses a value of $m = \min(\lfloor(x-1)(T-y)\rfloor, y)$. In this case, no integral value of $c$ can satisfy the constraints $c \leq (x-1)(T-y)$ and $m < c \leq y$ in case 3, so the resulting ratio is $\frac{T-(x-1)(T-y)}{T}$ or $\frac{T-y}{T}$, depending on whether or not $(x-1)(T-y) \leq y$.

In the other case, suppose the online algorithm chooses $m < \min(\lfloor(x-1)(T-y)\rfloor, y)$. The worst competitive ratios that the adversary can return are $\frac{T-m}{T}$ and $\frac{T-m}{y-m-1+x(T-y)}$. The latter expression can be written as $1 - \frac{(x-1)(T-y)-1}{y-m-1+x(T-y)}$. Because these are both decreasing functions of $m$, the online algorithm would choose $m=0$ to obtain a worst case ratio of $\frac{T}{y-1+x(T-y)}$.

The online algorithm, given $x,y,T$, can choose $m$ to obtain the better competitive ratio. We consider the cases to obtain a tight lower bound on the optimal competitive ratio. In the following analysis, let $r = y/T$ and $\alpha = 1-r$.

Case 1: $(x-1)(T-y) \geq y$, i.e. $x \geq \frac{y}{T-y} + 1 = \frac{1}{1-r}$ or $\alpha \geq 1/x$. Here, we also know that $\lfloor(x-1)(T-y)\rfloor > m \geq 0$, where $m$ is an integer, so $(x-1)(T-y) >1$, i.e. $x \geq \frac{1}{T-y}+1 \geq \frac{1}{T}+1$.  Note that the ratio our online algorithm can obtain is $\max(1-r, \frac{1}{r - 1/T + x(1-r)})$, where the first is a decreasing function of $r$ and the second is increasing. Thus, a lower bound on the max of the two ratios is obtained when we set them equal and solve for the optimal $\alpha' = 1/r'$, i.e.
\begin{align*}
\frac{1}{r'-1/T + x(1-r')} &= 1-r' \\
1 &= x{\alpha'}^2 + \alpha'(1-\alpha' - 1/T) \\
0 &= (x-1){\alpha'}^2 + (1-1/T)\alpha' - 1 \\
\alpha' &= \frac{-1 + 1/T + \sqrt{(1-1/T)^2 + 4(x-1)}}{2(x-1)}
\end{align*}
We confirm that this value of $\alpha'$ satisfies the condition $\alpha' \in (0,1)$, as $\alpha'$ is a decreasing function of $x$, and at the endpoint $x=1+\frac{1}{T}$, we can compute that
\begin{equation*}
\alpha' = \frac{-1 + 1/T + \sqrt{(1-1/T)^2 + 4(1/T)}}{2(1/T)} = \frac{-1+1/T+1+1/T}{2(1/T)} = 1
\end{equation*}
The value also satisfies $\alpha' \geq 1/x$ because
\begin {align*}
\frac{-1 + 1/T + \sqrt{(1-1/T)^2 + 4(x-1)}}{2(x-1)} &\geq 1/x \\
\iff (1-1/T)^2 + 4(x-1) &\geq (1-1/T)^2 + 4(1-1/x)(1-1/T) + 4(1-1/x)^2 \\
\iff 1 &\geq \frac{1}{x}(1-1/T) + \frac{x-1}{x^2} \\
\iff x + 1/x - 2 &\geq -1/T
\end{align*}
which is always true.

Then the lower bound is $1-r = \alpha'$ as given above, which asymptotically in $x$ is $\boxed{O(1/\sqrt{x})}$.

Case 2: $(x-1)(T-y) < y$, i.e. $x < \frac{1}{1-r}$ or $\alpha < 1/x$. Our online algorithm now obtains a ratio $\max(1-(x-1)(1-r), \frac{1}{r-1/T + x(1-r)})$. Both ratios are now increasing in $r$, so to lower bound the maximum, we set $r$ to minimum, or $\alpha$ to its maximum $1/x$.

We obtain the ratio
\begin{equation*}
\max\left(1-(x-1)\frac{1}{x}, \frac{1}{\frac{x-1}{x} + 1 - 1/T}\right)
= \max\left(1/x, \frac{x}{2x - 1 - x/T}\right)
\end{equation*}

Then since we are able to achieve the second ratio, we obtain asymptotically in $x$ a competitive ratio of $\boxed{\frac{1}{\left(2-\frac{1}{T}\right)}}$.

Surprisingly, we are able to obtain an equal asymptotic competitive ratio of $O(1/\sqrt{x})$ as in the arrival time unknown problem, given certain inputs, and a better competitive ratio of $1/(2-1/T)$ for any other input. 

\section{Generalizations}
In reality, the game of Mousehunt is a lot more complex. In order to better approximate the game, we can improve our model of the game and attempt to analyze those models. Thus, we can generalize our model further and ask the follow-up questions:
\begin{itemize}
\item What if there are traps $L_1, L_2 \dots L_m$ that can catch different mice and have different costs that are always available for purchase at a local trapsmith?
\item What if we don't know $x$, the effectiveness of Ronza's trap, ahead of time?
\item What if the mice give reward vectors $(p_i, g_i)$, where $p_i$ is the number of points gained and $g_i$ is the amount of gold collected from catching mouse $i$?
\item What if Ronza appears with multiple traps available? What if Ronza appears multiple times?
\item What if there is a probability distribution for the unknown variables in the problem? What can be said for some common distributions other than the uniform distribution? What if there are some other restrictions on the unknown variables, like upper and lower bounds?
\item What if multiple parameters are unknown at the same time? (For example, if both the cost and the arrival time of Ronza's trap are unknown).
\end{itemize}
These problems are much harder to analyze due to the increase in the number of parameters. For example, in the case of Ronza appearing multiple times, it becomes important to consider strategies that may buy certain Ronza traps early on in order to buy other Ronza traps in the future. Because these problems are better approximations of the actual game, they are definitely worth exploring in the future.

\section{Conclusion}
The above analysis shows many interesting results that spring out of a basic model of the game Mousehunt. Assuming the basic model of unknown arrival time, if we fix the probability distribution of the arrival times, we can model the problem as a Markov Decision Process problem and solve for the optimal deterministic algorithm. If we don't fix the probability distribution, an optimal deterministic online algorithm can still always achieve better than $1/\sqrt{x}$-competitive, and with randomization, can achieve better than $1/2$-competitive. Using Yao's minimax principle we show that we can do no better asymptotically than $1/2$-competitive over all possible values of the parameters, so the asymptotic bound of $1/2$-achieved by the randomized algorithm is strict.

Under the model of unknown cost, an optimal algorithm can actually achieve roughly a ratio of $O(1/\sqrt{x})$ in most cases and an even better ratio of $\frac{1}{\left(2-1/T\right)}$ in special cases. While these values are asymptotically similar to the values obtained for the model of unknown arrival time, the exact expressions for the competitive ratios and the analyses of these online algorithms differ greatly.

This Mousehunt problem appears similar to the well known ski rental problem, which has an optimal randomized online algorithm that achieves $(e-1)/e$-competitiveness. However, the algorithms and analysis results for these two problems are quite different, demonstrating the fundamental difference between the problems.

Finally, the Mousehunt problem has many more parameters, and so its optimal competitive ratio will depend on more parameters. For most parameter settings the competitive ratio is much better than the worst case ratio of $1/2$, but for the worst possible parameter settings, the ratio is still $1/2$.

The Mousehunt problem can also be applied to similar situations in other fields. For example, consider the problem of maximizing your net worth in life. At every time step, you have to decide between earning money now (e.g. working at a grocery store) and preparing yourself for future work by studying hard in school, establishing connections, and learning about entrepreneurship. In the case that a golden opportunity comes knocking at your door (for example, if a venture capitalist offers to hear your startup sales pitch), you could potentially get a huge boost in your pay rate if you impress him enough to have him invest in you. However, if you are unprepared and have not worked enough, then you can't take advantage of the opportunity when it comes. Our algorithm details when exactly it is better to earn money now or study now to prepare for the future, assuming that you have some estimate on the gain in wealth upon founding a startup.

While these results probably will not affect how people play the game of Mousehunt or how people try to maximize their net worth due to the level of complexity of these real situations compared to our model, they still show interesting results and methods of analysis of online algorithms, and how randomization can be used to improve an algorithm's competitiveness. Furthermore, this paper demonstrates how even a simple online problem can produce complex and unexpected analysis results, such as the asymptotic $1/\sqrt{x}$ competitive ratio.

\section{Acknowledgements}

We would like to thank our instructor Prof. Karger for guidance on this project, as well as HitGrab Inc. for developing Mousehunt.

\end{document}